\documentclass[a4paper,UKenglish,cleveref, autoref, thm-restate]{lipics-v2021}

\usepackage{amsmath}
\usepackage{amsthm}
\usepackage{amssymb}
\usepackage{mathrsfs}
\usepackage{marvosym}

\usepackage{multirow}
\usepackage{amsfonts}

\usepackage{textcomp}
\usepackage{enumerate}

\usepackage{float}
\usepackage[ruled,linesnumbered]{algorithm2e}
\title{A New Model for Massively Parallel Computation Considering both Communication and IO Cost} 

\author{Hengzhao Ma}{Shenzhen Institutes of Advanced Technology, Chinese Academy of Sciences}{hz.ma@siat.ac.cn}{https://orcid.org/0000-0002-2769-6138}{National Natural Science Foundation of China, grants 61732003, 61832003 and U1811461.}
\author{Xiangyu Gao}{School of Computer Science and Technology, Harbin Institute of Technology \and Shenzhen Institutes of Advanced Technology, Chinese Academy of Sciences}{gaoxy@hit.edu.cn}{}{}

\author{Jianzhong Li\footnote{Corresponding Author}}{Shenzhen Institutes of Advanced Technology, Chinese Academy of Sciences \and School of Computer Science and Technology, Harbin Institute of Technology}{lijzh@hit.edu.cn}{}{}

\author{Tianpeng Gao}{School of Computer Science and Technology, Harbin Institute of Technology \and Shenzhen Institutes of Advanced Technology, Chinese Academy of Sciences}{gaotp@hit.edu.cn}{}{}

\authorrunning{H. Ma, J. Li, X. Gao and T. Gao} 

\Copyright{Hengzhao Ma, Jianzhong Li, Xiangyu Gao and Tianpeng Gao} 



\category{} 

\relatedversion{} 

\EventEditors{John Q. Open and Joan R. Access}
\EventNoEds{2}
\EventLongTitle{42nd Conference on Very Important Topics (CVIT 2016)}
\EventShortTitle{CVIT 2016}
\EventAcronym{CVIT}
\EventYear{2016}
\EventDate{December 24--27, 2016}
\EventLocation{Little Whinging, United Kingdom}
\EventLogo{}
\SeriesVolume{42}
\ArticleNo{23}
\newcommand{\tabincell}[2]{\begin{tabular}{@{}#1@{}}#2\end{tabular}} 
\nolinenumbers
\begin{document}

\maketitle

\begin{abstract}
In the research area of parallel computation, the communication cost has been extensively studied, while the IO cost has been neglected. For big data computation, 
the assumption that the data fits in main memory no longer holds, and external memory must be used. Therefore, it is necessary to bring the IO cost into the parallel computation model. 
In this paper, we propose the first parallel computation model which takes IO cost as well as non-uniform communication cost into consideration. Based on the new model, we raise several new problems which aim to minimize the IO and communication cost on the new model. We prove the hardness of these new problems, then design and analyze the approximate algorithms for solving them. 
\end{abstract}

\section{Introduction}
In the research area of parallel computation, the communication cost has always been a  focus of attention, since it may well be the main bottleneck especially for the data intensive tasks such as SQL processing. 
In order to model the communication cost needed in parallel computation and minimize it, a lot of parallel computation models are proposed such as PRAM \cite{Karp1989}, BSP \cite{Valiant1990}, LogP \cite{Culler1996} and so on. In recent years there are several new modes that drew a lot of research attention, which are the Massively Parallel Computation model \cite{Karloff2010}, the Congest model \cite{Peleg2000}, and several variants of them. 

In the above mentioned models, it is assumed that the total data fit in the total main memory of a cluster. However, the assumption is no longer true for big data computation, and the external memory must be used. 
Thus, it is necessary to consider not only the communication cost but also the IO cost in parallel computation on big data. Here the IO cost refers to the cost of transferring the data from external memory to main memory. Unfortunately, existing parallel computation models rarely consider the IO cost. 

In this paper, a new parallel computation model is proposed which considers both communication and IO cost. 
In the following, we  first go over the existing models and discuss their disadvantages, and then introduce the motivation for proposing the new parallel computation model.

\subsection{Existing models and their disadvantages}
For the classical models such as PRAM, BSP and LogP, the disadvantage is that they have a too long history. The PRAM model remains a theoretical tool, but the BSP and LogP model are rarely considered in recent research works. 

In MPC model and its variants, the main consideration is the communication round complexity. There are only a few research results in database theory that care about the \textit{load} of each machine \cite{Afrati2010, Beame2014}, which is related to the IO cost. The situation is similar for the Congest and Congested Clique model. In these models the computation nodes are granted with unlimited computational power and unlimited memory, and the IO cost is ignored. In the next section we will explain the necessity for considering the IO cost in parallel computation.

What is worthy to mention is the topology-aware MPC model \cite{Hu2021}, which is the first model that considers different communication costs between different nodes. This is called topology awareness in \cite{Hu2021}. However, there exists a fatal drawback in that work, that is, the authors only consider the network to be the tree topology.

\subsection{The motivation of a new model}
1. \textbf{Why consider IO.}

The most new and important idea of the proposed model is to consider the IO cost in parallel computation. The reason to include IO cost is simple but somewhat ignored by former researchers, that is, the size of the data can not fit in the total main memory of a practical parallel computation system. For a typical MapReduce cluster \cite{Dean2008}, there may have around several tens of machines, and each machine has up to 128G memory. Hence the total main memory is in size of TB. However, there must be a future that the data to be processed grows into a size of PB or EB, but it is not realistic to have PB size of total main memory. In conclusion, it is necessary to consider external memory and bring IO cost into consideration in parallel computation. 

2. \textbf{Why not restricting external memory size.}

In this new model, we consider limited main memory but unlimited external memory.
In MPC model, the memory of a single machine is usually constrained to be sub-linear in the input size. The rationale is to rule out the trivial solution of transferring all the input to a single machine and solve the problem serially. However, this constraint is not appropriate when external memory is considered. The external memory of computers is becoming cheaper and larger these years. The largest storage capacity of a single disk have reached 20T. With up to 18 SATA ports on a mainboard, it is possible for a single machine to have 300T external memory. Thus, it is possible store the entire input in the external memory of a single machine, and the assumption of unlimited external memory is reasonable.

In another aspect, unlimited external memory is an analog with the classical Turing machine where the length of the working tape is unlimited. Just because of the unlimited length of the working tape of the Turing machine, analyzing the space complexity of algorithms on Turing machine is possible. Similarly, unlimited external memory makes analyzing the IO complexity possible.

In summary, for both realistic and research reasons, we can not make any constraint on the external memory of a single machine. Some may worry that unlimited external memory will lead to trivial single machine solution, but considering the single machine solution may not be optimal in IO cost, we can solve the optimization problem of minimizing IO cost and communication cost to avoid trivial single machine solution.

3. \textbf{Why consider topology awareness on a complete network.}

The idea of topology awareness is almost new \cite{Blanas2020, Hu2021}. The topology aware model captures the non-uniformity of the communication cost inside a cluster, i.e., the communication cost between different nodes are different. In \cite{Hu2021} the underlying communication network is modeled as a tree structure, which brings two disadvantages. First, the tree structure deviates from the mostly considered MPC and Congested Clique model which consider the network to be a complete graph. Second, the communication lower bounds derived in \cite{Hu2021} highly rely on the tree topology of the network. The basic idea is that data must transferred from one side to the other side for any edge in the network to complete the computation. This technique over-emphasizes the importance of the network topology but ignores the inherent complexity of the studied problem. And if the network topology is changed to a complete graph, all the lower bounds given in \cite{Hu2021} will be not applicable. Furthermore, the inherent communication complexity of the studied problem must be revealed to derive new bounds. 

4. \textbf{Why constant number of machines.}

We make another argument about the number of machines. In MPC and Congest model, the number of machines is considered to be a function of the input size \cite{Dean2008}, and thus can be arbitrarily large. However, this may not be realistic. A simple cluster for research purpose may consist of around 10 machines, and an enterprise cluster may have hundreds or thousands of machines.  Even though, thousands of machines are not comparable for data intensive tasks, where the input data is in PB of size and billions of records. Thus, assuming the number of machines is linear in the input size $n$ is not reasonable, even $n^{1-\epsilon}$ is not reasonable neither. The appropriate choice is to set the number of machines to be a constant.

Many theoretical models assume that there are arbitrary number of processors. To make the new model compatible with these theoretical models, it can be considered that there are arbitrarily many virtual processors, and these virtual processors must be mapped to the constant number of physical machines. 
This is natural in applicational environment. We will prove that the new model with constant number of machines and unlimited external memory can simulate the  the PRAM model in a reasonable time. 

5. \textbf{Comparing the proposed model with existing ones.}

In Table \ref{tab:model-comparison} we list some parallel computation models for comparison, together with the new model proposed in this paper, which is called EMPC. The EMPC model is not only the first parallel computation model that involves IO cost, but also has a lot of differences and novelty compared to the existing models. The details of the EMPC model will be elaborated in Section \ref{sec:model}.

\begin{table}
	\resizebox{\textwidth}{!}{
		\begin{tabular}{|c|c|c|c|c|c|c|}
			\hline
			\multirow{2}{*}{Model}& \multirow{2}{*}{Network Topology} & \multirow{2}{*}{\tabincell{c}{Communication\\Non-uniformity}} & \multicolumn{2}{|c|}{\tabincell{c}{Communication Cost\\ Considered?}} & \multirow{2}{*}{\tabincell{c}{IO Cost\\ Considered?}} &
			Restrictions\\ 
			\cline{4-5}
			& & & Amount & Round & &\\
			\hline
			TA-Trees & Trees & Yes & Yes  & - & No & Tree network topology\\
			\hline
			MPC & Complete Graph & No & - & Yes & Yes & Total main memory \\
			\hline
			\tabincell{c}{Congested\\ Clique} & Complete Graph & No & - & Yes & No& Communication bandwidth \\
			\hline
			Congest & General Graph& No & - & Yes & No &Communication bandwidth \\
			\hline
			\tabincell{c}{EMPC\\(This paper)} & Complete Graph & Yes & Yes & - & Yes& \tabincell{c}{Limited main memory\\Unlimited external memory} \\
			\hline
		\end{tabular}
	}
	\caption{Comparison of the models}\label{tab:model-comparison}
\end{table}

\subsection{Declaration of New Research Problems}
By considering both IO cost and non-uniform communication cost, a lot of new research problems emerge on the proposed EMPC model. Here we declare the problems to be solved in this paper.

\textbf{IO-optimality in parallel computation.} By introducing IO cost into the model, we define the IO-optimality which is an analog with the total work optimality in classic analysis of parallel algorithms \cite{Kruskal1990}. Then we choose some state-of-art algorithms on MPC and Congest Clique model, and determine the IO-optimality of them. We will see that some algorithms are optimal, some non-optimal, some even super-optimal. 

\textbf{The data redistribution problem.} This problem arises merely by changing the underlying network topology of \cite{Hu2021} from trees to full connected graphs. We have claimed that the communication lower bounds given in \cite{Hu2021} is not applicable for full connected networks. Using the parallel sorting as an introducing example, we define a new problem called Data Redistribution Problem (DRP). The goal of DRP is to find a best plan to redistribute the data so that the resulted data distribution can produce correct result, and the communication cost of the data redistribution is minimized. We will see that DRP is NP-complete while the IO cost is not yet considered.

\textbf{Minimizing total cost of communication and IO.} The ultimate problem in the new model is to minimize total cost, which is the communication cost plus the IO cost. Also using the parallel sorting as an example, we show that the optimization problem considering both communication and IO cost is in the XP class. Then we give an approximate algorithm and prove the approximation ratio.

\subsection{Organization}\label{subsec:organize}
The rest of this paper is organized as follows. In Section \ref{sec:model} we describe the proposed model in detail and prove the computational power of the new model.  In Section \ref{sec:optmize-io} we  define the IO-optimality in parallel computation and analyze the IO cost of several existing algorithms. Section \ref{sec:optimize-comm} deals with the problem of optimizing the communication cost on the new model. Section \ref{sec:optimize-comm-io} considers the problem of optimizing the communication and IO cost simultaneously, which is the most important problem on the proposed model. Finally Section \ref{sec:conc} concludes this paper and we declare some future works in Section \ref{sec:fwork}.

\section{The Proposed Model}\label{sec:model}
The proposed model is called EMPC, which stands for Enhanced Massively Parallel Computation. Basically the model can be regarded as a MPC model with topology awareness and IO concern. There are constant number of $p$ machines connected with a weighted complete graph, where the weights on the edge represent the communication cost between the machines and are known parameters of the model. We use $C[i,j], 1\le i,j\le p$, to denote the communication cost matrix. Note that $C[i,i]=0$ for $1\le i\le p$. We make no assumption on symmetric communication, i.e., $C[i,j]$ may not equals to $C[j,i]$. See Figure \ref{fig:network-costs} for a demonstration.

There is also a parameter $C_{IO}$ representing the cost of a single IO operation. In this paper we always assume $C_{IO}=1$ for simplicity. The total cost is modeled as the addition of communication and IO costs.

\begin{figure}[H]
	\caption{The network topology and communication costs}\label{fig:network-costs}
	\centering
	\begin{subfigure}{0.3\textwidth}
		\caption{The cluster network}
		\includegraphics[width=0.73\textwidth]{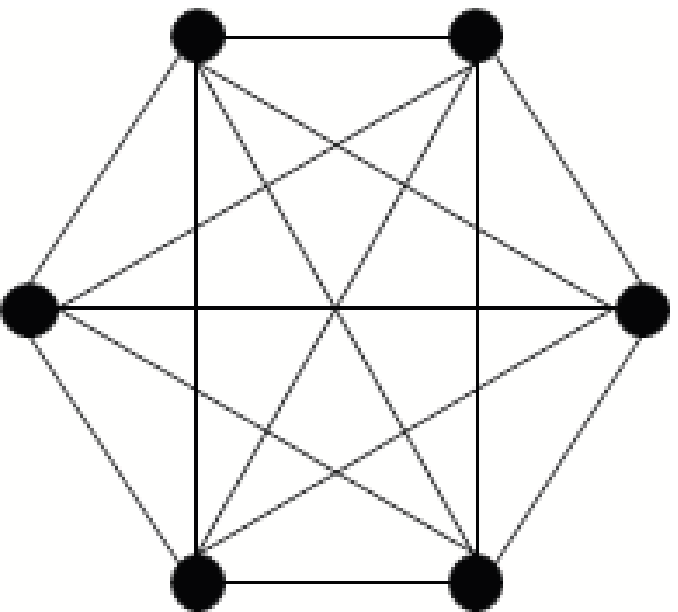}
	\end{subfigure}
	\begin{subfigure}{0.69\textwidth}
		\caption{The communication cost matrix}
		\includegraphics[width=0.8\textwidth]{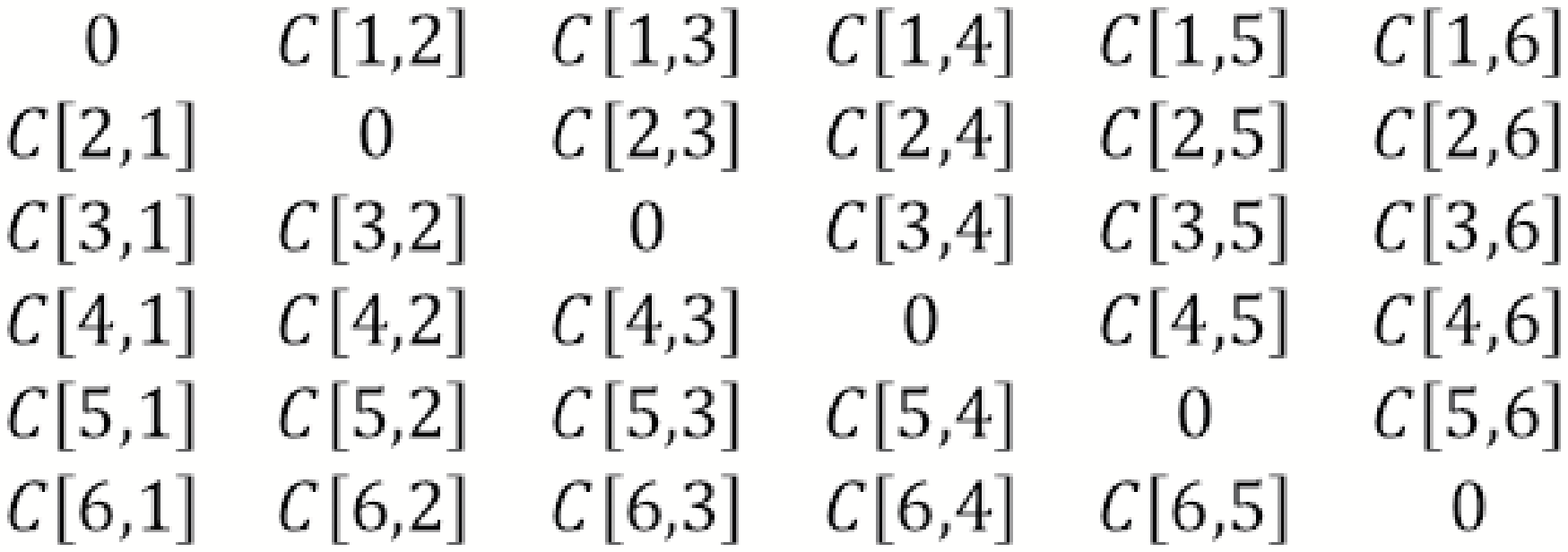}		
	\end{subfigure}

\end{figure}

The computation on EMPC proceeds in synchronous rounds. In each round, the machines first conduct local computation, then exchange messages for the next round of computation. This behavior is the same with the classic MPC model. 
But notice that the local computation on EMPC may involve IO operations. Data must be read to main memory for conducting computation tasks. The communication messages must be exchanged through the main memory but not the external memory. See Figure \ref{fig:framework} for a demonstration of the computation framework.   

\begin{figure}[H]
	\centering
	\caption{The framework of parallel computation on EMPC}\label{fig:framework}
	\includegraphics[width=0.5\textwidth]{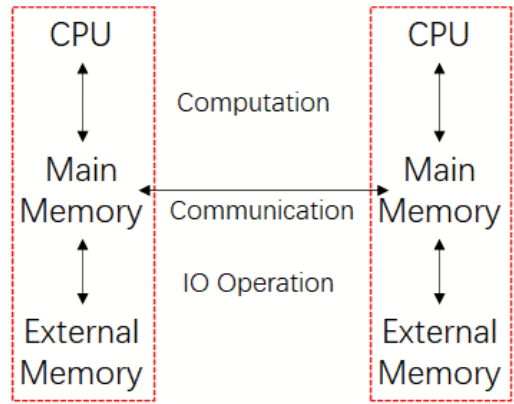}
\end{figure}

The size of the main memory of each machine is limited but the external memory is unlimited. There is no limit on the number and length of communication messages between the machines in each round, since the messages which exceeds the size of the main memory can be spilled to the external memory.

Note that the classical IO complexity usually involves a parameter $B$ which denotes the disk block size. In our model the parameter $B$ is ignored for simplicity. For example, the IO complexity of scanning is $scan(n)=O(n/B)$, but we ignores the parameter $B$ and thus the IO complexity of scanning on EMPC is $scan(n)=O(n)$. This is more like considering the load complexity on the MPC model \cite{Beame2014}.

\subsection{Computational Power of the New Model}\label{subsec:model-power}
In this section we show the computation power of the proposed EMPC model by simulating MPC and PRAM on EMPC.

\begin{theorem}\label{thrm:EMPC-MPC-simulation}
	Any $T(n)$ round algorithm on MPC model using $P(n)$ processors and $M(n)$ memory on each processor can be simulated on EMPC with constant number of $p$ machines in $T(n)$ rounds, using $O(M(n))$ memory and  $O(M(n)\cdot P(n))$ total external memory .
\end{theorem}

\begin{proof}
	First we assign the $P(n)$ processors in MPC model to $p$ machines in EMPC model, and allocate the external memory space for the processors. Each machine is assigned with $P(n)/p$ processors, and the total external memory usage is $M(n)\cdot P(n)$. We assume the space needed to record the state of a processor on external memory is a constant. Now it only needs to simulate $P(n)/p$ processors in MPC on each machine in EMPC, which is an easy imitation of the process scheduling method on modern operating systems. The details are omitted.
\end{proof}

\begin{theorem}\label{thrm:EMPC-PRAM-simulation}
	Any CREW PRAM algorithm using $M(n)$ total memory and $P(n)$ processors can be simulated on EMPC with constant number of $p$ machines, using at least $O((\log{P(n)}+\log{M(n)})$ total main memory and at most $O((\log{P(n)}+\log{M(n)})(P(n)+M(n)))$ total external memory. If the CREW PRAM algorithm runs in $t$ time, then the simulated algorithm on EMPC runs in $2t$ rounds.
\end{theorem}

\begin{proof}
	In this proof we assume that the word length of PRAM and EMPC are constant, and the space needed to record the state of a processor on external memory is also constant. 
	
	To start the simulation, we first assign the processors in PRAM with an unique $pid$ and assign the machines in EMPC with a unique $mid$, then map the $P(n)$ processors to $p$ machines. The mapping is represented as a list of  $<pid, mid>$ records and is known to all machines.

	The shared memory in PRAM is mapped to a section of external memory on EMPC, which is called working external memory for reference. Since the size of the external memory is unlimited on EMPC, the whole section of working external memory can be assigned to a single machine. Yet for a more rigorous proof we partition the working external memory across the $p$ machines. The partition is recorded as a $<addr, mid>$ list and the list is known to all machines.

	Now we propose the simulation algorithm with least possible main memory usage.
	\begin{enumerate}
		\item Each machine reads its $pid$ lists, and for each $pid$, output a record $<W, pid, addr, content>$ or $<R, pid, addr>$ to external memory based on its action. The $W/R$ stands for writing or reading operation. The record is noted as $<W/R, pid, addr, content/\emptyset>$ for simplicity. 
		\item Each machine reads the list of $<W/R, pid, addr, content/\emptyset>$ records, searches the $<addr,mid>$ mapping, and outputs a record $<W/R, pid, mid, addr, content/\emptyset>$ to external memory. The $<W/R, pid, addr, content/\emptyset>$ list can be discarded.
		\item Each machine reads the list of $<W/R, pid, mid, addr, content/\emptyset>$ records. For any record with $mid$ different from the $mid$ of itself, send a message $<W/R, mid', pid, addr, content/\emptyset>$ to machine $mid$, where $mid'$ is the $mid$ of the sender. The messages are written to external memory upon receiving. Note that in this step the communication and IO operation are concurrently executed.
		\item Each machine reads the messages list $<W/R, mid', pid, addr, content/\emptyset>$. For writing operations, the content field in the message is written to the working external memory at address $addr$. For reading operations, the content at address $addr$ in working external memory is read, and a message $<mid',pid, addr, content>$ is sent to machine $mid'$. The messages are written to external memory upon receiving.
		\item Each machine read the message list $<mid',pid, addr, content>$ and change the state of processor $pid$.
	\end{enumerate}

	The correctness of the above simulation process is straightforward. It can be noticed that the simulation consists of 2 rounds of local computation and communication, and thus the round complexity follows.
	
	Next we calculate the size of the  main memory and external memory needed in the simulation.

	\begin{itemize}
		\item Working external memory: $M(n)$ size.
		\item The $pid$ list on each machine: each $pid$ is of length $O(\log{P(n)})$, and the thus all the $pid$'s needs $O(P(n)\log{P(n)})$ external memory.
		\item $<addr, mid>$ mapping: each $addr$ is of length $O(\log{M(n)})$, and each $mid$ is of constant length. Then the total list of $<addr, mid>$ is of size $O(M(n)\log{M(n)})$. Since the $<addr, mid>$ must be duplicated for each machine, the total external memory usage of $<addr, mid>$ mapping is $O(M(n)\log{M(n)}\cdot p)= O(M(n)\log{M(n)})$. 
		\item $<pid,mid>$ mapping:  $O(P(n)\log{P(n)})$. From now on the detailed calculation is omitted and only the result of external memory usage is presented.
		\item $<W/R, pid, addr, content/\emptyset>$ list: $O((\log{P(n)}+\log{M(n)})P(n)  )$.
		\item $<W/R, pid, mid, addr, content/\emptyset>$ list: $O((\log{P(n)}+\log{M(n)})P(n)  )$.
		\item $<W/R, mid', pid, addr, content/\emptyset>$ message list: $O((\log{P(n)}+\log{M(n)})P(n)  )$
		\item $<mid',pid, addr, content>$ message list: $O((\log{P(n)}+\log{M(n)})P(n)  )$
	\end{itemize}

	In summary, the total external memory usage is $O((\log{P(n)}+\log{M(n)})(P(n)+M(n))  )$. A constant $p$ which is the number of machines is hidden.
	
	As for the main memory usage, at any time only one record of each the above list needs to be maintained in main memory. Thus the main memory usage is of size $O((\log{P(n)}+\log{M(n)})$.

	By now the theorem has been proved.
\end{proof}

Here we have several remarks. First, Theorem \ref{thrm:EMPC-MPC-simulation} and \ref{thrm:EMPC-PRAM-simulation} show that the EMPC model acquires equivalent computational power using constant number of machines and unlimited external memory, in comparison with arbitrary number of machines and limited main memory in MPC and PRAM. Second, in the proof of Theorem \ref{thrm:EMPC-PRAM-simulation} we try to use main memory as small as possible, and use external memory as large as possible, and thus the simulation incurs a lot of IO operations. Actually, the external memory usage can be reduced. A simplest way is to use a modular mapping for the $<pid,mid>$ mapping, i.e., map each $pid$ to the machine ($pid$ mod $p$) where $p$ is the number of machines, and thus the $<pid,mid>$ list can be spared. Similar technique can be used on the $<addr, mid>$ mapping. On the other hand, the main memory usage can be raised, and thus the external memory usage and number of IO operations can be reduced. For example, it is reasonable that each machine has enough memory to store all the messages in each round. But note that the size of the main memory must be reasonably low, otherwise the external memory would be meaningless. Finally, we note that the communication and IO operation on EMPC model can be overlapped.

\section{Optimizing IO: The Parallel IO-optimality}\label{sec:optmize-io}
Analog to the definition of work-optimality in classic parallel computation \cite{Kruskal1990}, we here define the parallel IO-optimality on the EMPC model. For a given problem $\mathcal{P}$, denote $IO(\mathcal{A}_P)$ as the total amount of IO operations of a parallelized algorithm $\mathcal{A}_P$ executed on $p$ machines, and $IO(\mathcal{A}_S)$ as the amount of IO operations of a serialized algorithm $\mathcal{A}_S$. We have the following definitions.

\begin{definition}[super-IO-optimal]
	An algorithm $\mathcal{A}_P$ running on the EMPC model is called super-IO-optimal against a serial algorithm $\mathcal{A}_S$ if $IO(\mathcal{A}_P)\le IO(\mathcal{A}_S)$. If $\mathcal{A}_P$ is super-IO-optimal against any serial algorithm, $\mathcal{A}_P$ is called super-IO-optimal for problem $P$.
\end{definition}

\begin{definition}[IO-optimal]
	An algorithm $\mathcal{A}_P$ running on the EMPC model with $p$ machines is called IO-optimal against a serial algorithm $\mathcal{A}_S$ if $IO(\mathcal{A}_P)\le O(1)\cdot IO(\mathcal{A}_S)$. If $\mathcal{A}_P$ is IO-optimal against any serial algorithm, $\mathcal{A}_P$ is called IO-optimal for problem $P$.
\end{definition}

\begin{definition}[non-IO-optimal]
	An algorithm $\mathcal{A}_P$ running on the EMPC model is called non-IO-optimal for problem $P$ if there exists a serial algorithm $\mathcal{A}_S$ such that $IO(\mathcal{A}_P)=\Omega(IO(\mathcal{A}_S))$.
\end{definition}

In this section we will use the Minimum Spanning Tree (MST), Maximum Matching and Sorting as the demonstrative examples and determine the IO-optimality of the state-of-art algorithms for these problems.

\subsection{Minimum Spanning Tree}
The state-of-art parallel algorithm for MST with the lowest round complexity is given in \cite{Nowicki2021}, which is a deterministic algorithm using $O(1)$ rounds on Congested Clique model. By the simulation algorithm given in \cite{Hegeman2015}, the algorithm can be simulated on MPC using also $O(1)$ rounds. 

In this section we will simulate the algorithm given in \cite{Nowicki2021} on EMPC model with $O(n)$ main memory and unlimited external memory, and compare the IO cost of this algorithm to the classical Kruskal algorithm for MST. Note that the $O(n)$ main memory usage is the same with the algorithm in \cite{Nowicki2021}.

The algorithm  given in \cite{Nowicki2021} is split into several sub-procedures. We cite the part related to our following proofs, which is referred as  Algorithm \ref{alg:nowicki-reduce}.

\begin{algorithm}
	\caption{ReduceToSparse-Nowicki}\label{alg:nowicki-reduce}
	\KwIn{a weighted graph G}
	\KwOut{G', a weighted graph with $O(n)$ edges and $O(m+n)$ edges}
	$G^b=(V^b,E^b)\gets$ InitialReduction(G) \;
	Let $V_1^b,\cdots,V_{m/n}^b$ be a partition of the vertices of $V^b$ into disjoint subsets of size $\Theta(n^2/m)$\;
	Partition the edges in such a way that for all $i\le j$ the edges $E_{i,j}^b=\{(u,v)\mid u\in V_i^b,v\in V_j^b, (u,v)\in E^b \}$ are in the memory of a single processor\;
	For all $i\le j$ compute the minimum spanning forest $F_{i,j}^b$ of the graph $(V_i^b\cup V_j^b, E_{i,j}^b)$\;
	Let $G'=(V^b,\bigcup_{i,j}\{edges\; of\; F_{i,j}^b\})$\; 
	Let $V_i'=\sum_{j=(i-1)\sqrt{m/n}+1}^{i\sqrt{m/n}} V_j^b$\;
	Partition the edges in such a way that for all $i\le j$ the edges $E_{i,j}'=\{(u,v)\mid u\in V_i',v\in V_j', (u,v)\in E' \}$ are in the memory of a single processor\;
	For all $i\le j$ compute the minimum spanning forest $F_{i,j}'$ of the graph $(V_i'\cup V_j', E_{i,j}')$\;
	Return the graph $(V^b,\bigcup_{i,j}\{edges\; of\; F_{i,j}'\}))$ with the partition $V_1',V_2'\cdots, V_{\sqrt{m/n}}'$.
	
\end{algorithm}

\begin{theorem}
	If Algorithm \ref{alg:nowicki-reduce} is implemented on the EMPC model with $O(n)$ memory and unlimited external memory, the amount of IO operations is at least $O(m^2/n)$.
\end{theorem}

\begin{proof}
	The proof will focus on the third line of the cited algorithm: partition the edges in such a way that for all $i\le j$ the edges $E_{i,j}^b=\{(u,v)\mid u\in V_i^b,v\in V_j^b, (u,v)\in E^b \}$ are in the memory of a single processor.
	
	According to the analysis in \cite{Nowicki2021}, the size of each $E^b_{i,j}$ is $O(n)$, and the algorithm intends to store each $E^b_{i,j}$ in the memory of a single processor. This operation is feasible if the number of processor is $O(n)$. However, in EMPC model there are only constant number of $p$ machines. Since the main memory of each machine is $O(n)$, the edges in $E_{i,j}^b$ must be stored in external memory, and when computing $F_{i,j}^b$ the edges in $E_{i,j}^b$ must be read from external memory using IO operations. Since each $E_{i,j}^b$ must be read once for each $1\le i\le j\le m/n$, it can be verified that each edge is read for $m/n$ times, and the total amount of external memory read is $O(m^2/n)$. 
\end{proof}

\begin{theorem}
	To implement the Kruskal algorithm on a single machine with $O(n)$ main memory and unlimited external memory, the amount of IO operations is $sort(m)+m=O(m\log_n{m})$.
\end{theorem}

\begin{proof}
	The Kruskal algorithm requires sorting the edges according to edge weight, and then a linear scanning of the sorted edges to compute the MST. Then the above results can be immediately derived.
\end{proof}

\begin{theorem}
	The parallel algorithm given in \cite{Nowicki2021} is non-IO-optimal for computing MST.
\end{theorem}

\begin{proof}
	The claim follows from $m^2/n=\Omega(m\log_n{m})$ and the definition of non-IO-optimality.
\end{proof}

\subsection{Maximum Matching}
For maximum matching we choose the algorithm given in \cite{Ghaffari2018} as the demonstrative example. In \cite{Ghaffari2018} the authors first proposed a serialized algorithm, and transformed it into a MPC implementation. We will analyze the performance of the serialized and parallelized algorithm in  \cite{Ghaffari2018} on the EMPC model. For more parallel algorithms to compute maximum matching please refer to the references in \cite{Ghaffari2018}.

\begin{algorithm}
	\caption{MM-Ghaffari-Serialized}\label{alg:MM-Ghaffari-serial}
	\KwIn{unweighted graph $G=(V,E)$}
	\KwOut{an approximated fractional maximum matching}
	For each edge $e\in E$, set $x_e=1/n$\;
	\While{exists unfrozen edges}{
		Freeze each vertex $v$ for which $y_v=\sum_{e\in v}{x_e}\ge 1-2\epsilon$ and freeze all its edges\;
		For each active edges, set $x_e\gets x_e/(1-\epsilon)$\;
		
	}
	Output the values $x_e$ as a fractional matching\;
\end{algorithm}

\begin{lemma}
	Algorithm \ref{alg:MM-Ghaffari-serial} can be implemented such that the number of IO operations equals to the current number of active edges in each iteration.
\end{lemma}

\begin{proof}
	We implement the algorithm on a single machine with the following initial status. The $m$ edges are stored in external memory in the form of adjacent list, together with its edge weight initially set to $1/n$. The $n$ vertexes are kept in main memory with some auxiliary information, e.g. active or not. For each vertex we store a pointer to the start position of its corresponding adjacent list. In each iteration the edges of the active vertexes are scanned into main memory. After testing whether the vertex is to be frozen, the weight of active edges are updated and wrote back to external memory. The frozen edges of active vertex are switched to the end of its adjacent list.
	
	According to the above implementation, the number of IO operations in each iteration exactly equals to the number of active edges.
\end{proof}

We only cite the parallelized part of the algorithm given in \cite{Ghaffari2018}. The other codes for serialized processing are omitted.

\begin{algorithm}
	\caption{MM-Ghaffari-Parallelized}\label{alg:MM-Ghaffari-parallel}
	$\cdots$(omitted)\;
	Set \# of machines $m=\sqrt{d}$, \# iterations $I=\frac{\log{m}}{10\log{5}}$\;
	Partition $V'$ into $m$ sets $V_1,\cdots, V_m$ by assigning each vertex to a machine independently and uniformly at random\;
	\For{ each $i\in \{1,\cdots, m\}$ in parallel execute $I$ iterations}{
		For each $v\in V_i$, estimate its weight and freeze it if necessary\;
		For each active edge, update its estimated weight\;
		Increment the total iteration count: $t\gets t+1$;
	}
	Update $d\gets d(1-\epsilon)^I$\;
	Update the legal vertex set $V'$\;
	$\cdots$(omitted)\;	
\end{algorithm}

\begin{lemma}
	If Algorithm \ref{alg:MM-Ghaffari-parallel} is to be implemented on EMPC model with $O(n)$ memory and unlimited external memory, the number of IO operations needed in each iteration equals to the current number of active edges.
\end{lemma}

\begin{proof}
	This can be easily verified  according to Algorithm \ref{alg:MM-Ghaffari-parallel}.
	
\end{proof}

\begin{theorem}
	If implemented on EMPC, the parallelized algorithm \ref{alg:MM-Ghaffari-parallel} is IO-optimal against the serial algorithm \ref{alg:MM-Ghaffari-serial}.
\end{theorem}

\begin{proof}
	Although the parallelized algorithm uses the round compression technique to reduce the number of phases, the total number of iterations remains the same. We have proved that in each iteration, both the parallelized and serial algorithm need a number of IO operations which equals to the current number of active edges, and the number of active edges are the same according to the analysis in \cite{Ghaffari2018}. Thus the number of IO operations of the parallel algorithm equals with the serial algorithm. By the definition of IO-optimality, Algorithm \ref{alg:MM-Ghaffari-parallel} is IO-optimal against Algorithm \ref{alg:MM-Ghaffari-serial}.
\end{proof}

\subsection{Sorting}
For the parallel sorting we use the following Algorithm \ref{alg:terasort} known as TeraSort.

\begin{algorithm}
	\caption{TeraSort}\label{alg:terasort}
	\KwIn{A colloction of unsorted data distributed across $p$ machines}
	\KwOut{Sorted collection of the data}
	\tcc{Phase 1: Sample and split}
	
	Each machine sample the local data by the method given in \cite{Tao2013}\;
	Each machine send the sample to machine $M_1$\;
	$M_1$ sort the collected sample in main memory, and compute the splitters\;
	$M_1$ broadcast the splitters\;
	\tcc{Phase 2: Redistribute and pre-sort}

	The machines redistribute the data according to the splitters, sending the data in the $i$-th interval to machine $M_i$\;
	When receiving data, each machine stores the data in main memory and sort it; If the main memery buffer is full, flush it to external memory\;
	\tcc{Phase 3: Local sort}
	When the redistribution is finished, each machine sort the local data by merging the sorted runs\;
\end{algorithm}

\begin{theorem}
	The TeraSort is super-IO-optimal for the sorting problem.
\end{theorem}

\begin{proof}
	In the following proof $n$ denotes the size of the input data to be sorted, and $m$ denotes the size of the main memory of each machine. We assume that $m=o(n)$.
	For Algorithm \ref{alg:terasort}, Phase 1 needs totally $O(|sample|)$ IO operations if the sampling process is conducted by fist generating random numbers in main memory and retrieving the sampled records from external memory. Since the size of the sample fits in the main memory, we have $O(|sample|)=O(m)$. Phase 2 needs 0 IO since all computation occurs in main memory. Phase 3 needs $O(\frac{n}{p}\log_m{\frac{n}{p}}\cdot p)=O(n\log_m{\frac{n}{p}})$ IO operations in expectation. The total IO cost of Algorithm \ref{alg:terasort} is $O(m+n\log_m{\frac{n}{p}})$. On the other hand, the known lower bound of IO operations for sorting on external memory on a single machine is $O(n\log_m{n})$. Now we compare $O(m+n\log_m{\frac{n}{p}})$ and  $O(n\log_m{n})$ as follows.
	$$m+n\log_m{\frac{n}{p}}=n\log_m{((\frac{m}{n})^m\frac{n}{p})}= n\log_m{(o(1) \cdot \frac{n}{p})  }=n\log_m(o(1)\cdot n) $$

	Since the IO cost of parallel sorting, $O(m+n\log_m{\frac{n}{p}})$, is less than $O(n\log_m{n})$, we reach the conclusion that the TeraSort is super-IO-optimal for the sorting problem.
\end{proof}

\subsection{A conclusive remark}
In this section we found that the state-of-art theoretical parallel algorithm for MST with the best round complexity is actually non-IO-optimal. The lesson is that the theoretical algorithms tend to assume that there are many processors, usually as a function of the input size. These processors must be mapped to a constant number of physical machines, which immediately causes severe IO cost.

\section{Optimizing Communication: The Data Redistribution Problem}\label{sec:optimize-comm}
We use sorting as the example to introduce the data redistribution problem. The pseudo code for parallel sorting is already given in Algorithm \ref{alg:terasort}. In this section we pay attention to Phase 2, the redistribution phase. The pseudo code asks the data in the $i$-th interval to be sent to the $i$-th machine. But this may not be optimal in terms of communication cost.
Consider the following extreme case. The data is initially reversely ordered from machine $1$ to $p$, i.e., machine $M_p$ has the smallest data records, and $M_1$ has the largest data records. In this case we only need to reassign the order of the machines to acquire full ordered data. But the algorithm asks to redistribute the data entirely, causing large amount of unnecessary communication.

Thus, we point out an important problem which is ignored before, that is, how to assign the order from $1$ to $p$ to the machines so that the total data communication is minimized. We define the following data redistribution problem (DRP) under the EMPC model.

\begin{definition}[Data Redistribution Problem, DRP]
	Input: A $n\times n$ transmission matrix $T[i,j]$ and a $n\times n$ cost matrix $C[i,j]$. $T[i,j]$ represents the amount of data initially residing in physical machine $i$ to be transferred to virtual machine $j$. $C[i,j]$ represents the cost to transfer one unit of data from physical machine $i$ to $j$.
	
	Output: find a permutation $\pi$ of $[1,n]$, which represents assigning virtual machine $i$ to machine $\pi_i$, such that the following total communication cost is minimized:
	$$\min_{\pi\in \Pi(n)} \sum_{i=1}^n\sum_{j=1}^n T[i,j]C[i,\pi_j] $$
	where $C[i,i]=0$ for $1\le i\le n$.
\end{definition}

\begin{remark}\label{remk:drp-topology-important}
	If the non-zero elements in $C[i,j]$ are all set to $1$, which means that we ignore the non-uniformity of the communication in the cluster and consider the classic MPC model, the problem degenerates to the linear assignment problem which can be solved using Huangarian algorithm in polynomial time. This indicates that it is critical to consider the topology awareness in parallel computation. 
\end{remark}

\begin{remark}\label{remk:drp-qap}
	If we do not consider the initial data distribution, the problem will be the same as the well known quadratic assignment problem (QAP), which is NP-complete \cite{Sahni1976}. QAP can be regarded as assigning $p$ virtual machines to $p$ physical machines, where there are communication tasks between virtual machines and communication costs between physical machines. The difference between QAP and DRP is that in DRP the communication tasks are defined between a physical machine and a virtual machine, which makes DRP a little easier that QAP.
\end{remark}

According to Remark \ref{remk:drp-topology-important} and \ref{remk:drp-qap}, it can be concluded that the DRP is a new problem under the EMPC model whose hardness seems to be sandwiched between the linear assignment problem and the quadratic assignment problem. The following theorem shows that DRP is NP-complete.

\begin{theorem}\label{thrm:drp-np}
	DRP is NP-complete.
\end{theorem}

\begin{proof}
	We prove the NP-completeness by reducing the problem of TSP on a full bi-party graph (TSP-FB) \cite{Frank1998} to DRP. The definition of TSP-FB is as follows.
	Let $V_1=\{v_{1,1},v_{1,2},\cdots,  v_{1,n}  \}$ and $V_2=\{v_{2,1},v_{2,2},\cdots, v_{2,n}\}$ be two vertex sets with $|V_1|=|V_2|=n$. Let $E=\{(v_{1,i},v_{2,j})\mid 1\le i\le n, 1\le j\le n \}$ be the edge set.  Let $\mathcal{W}: E\rightarrow \mathcal R$ be the edge weights. Then $(V_1,V_2,E,\mathcal{W})$ form a weighted full bi-party graph. Denote $w_{i,j}$ as the weight of edge $(v_{1,i},v_{2,j})$.
	Given a weighted full bi-party graph $G=(V_1,V_2,E,\mathcal{W})$ and a number $k$, decide whether there exists a Hamilton cycle in $G$ with total weight no more than $k$.
	
	For any given instance of TSP-FB, we construct an instance of DRP as follows.
	
	Let $T[i,j]$ and $C[i,j]$ be two matrices with dimension $n$. Let $C[i,j]=w_{i,j}$. Let $T[i,j]$ be as follows.
	
	For odd $n$:
	\begin{equation}
		T[i,j]=\left\{
		\begin{aligned}
			1,\; & j=1+(i+1)\; mod\; n \; or\; j=1+(i-1)\; mod\; n \\
			0,\; & otherwise \\
		\end{aligned}
		\right.
	\end{equation}
	For even $n$:
	\begin{equation}
		T[i,j]=\left\{
		\begin{aligned}
			1,\; & (i=1,\; j=2)\; and\; (i=1,\; j=n-1) \\
			1,\; & (i=2,\; j=1)\; and\; (i=2,\; j=n) \\
			1,\; & (i\ne 1, \; i\; is\; odd)\; and \;(j=1+(i-1)\; mod\; n \; or\; j=1+(i-3)\; mod\; n)\\
			1,\; & (i\ne 2,\; i\; is\; even)\;and \;( j=1+(i-1)\; mod\; n\; or\; j=1+(i+1)\; mod\; n)\\
			0,\; & otherwise
		\end{aligned}
		\right.
	\end{equation}
	
	The formation of the matrix $T$ corresponds to the following graph.
	
	\begin{figure}[H]
		\centering
		\begin{subfigure}{0.45\textwidth}
			\caption{odd case}
			\includegraphics[width=0.3\textwidth]{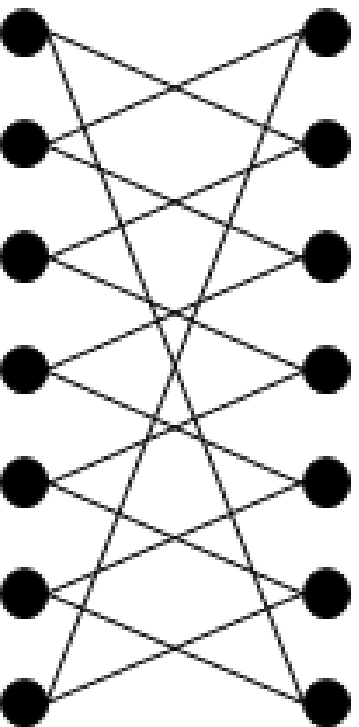}
		\end{subfigure}
		\begin{subfigure}{0.45\textwidth}
			\caption{even case}
			\includegraphics[width=0.3\textwidth]{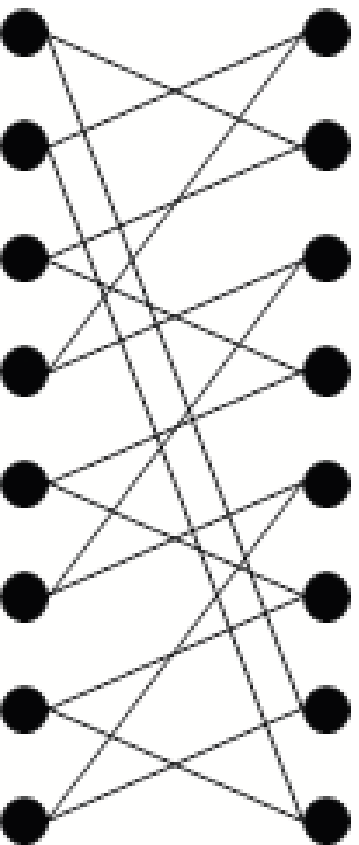}
		\end{subfigure}

	\end{figure}
	
	Suppose there exists an Hamilton cycle $HC^*=(v_{1,i_1}, v_{2,j_1}, v_{1,i_2}, v_{2,j_2},\cdots v_{1, i_n},v_{2,j_n})$ with total weight less than $k$.
	The Hamilton cycle corresponds to the following permutation, for odd and even cases, respectively.
	
	\begin{figure}[H]
		\centering

		\begin{subfigure}{0.45\textwidth}
			\caption{odd case}
			\includegraphics[width=0.77\textwidth]{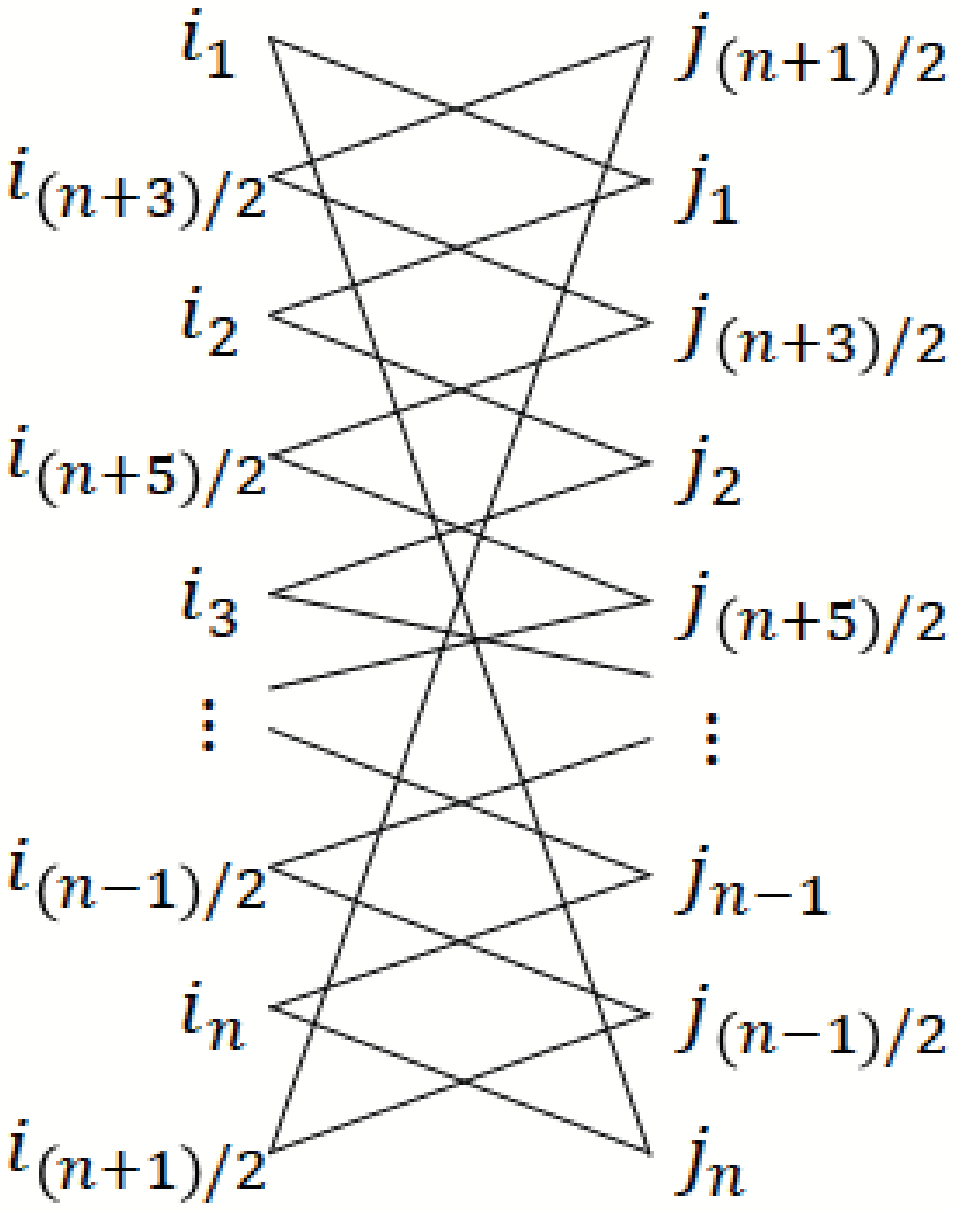}
		\end{subfigure}
		\begin{subfigure}{0.45\textwidth}
			\caption{even case}
			\includegraphics[width=0.7\textwidth]{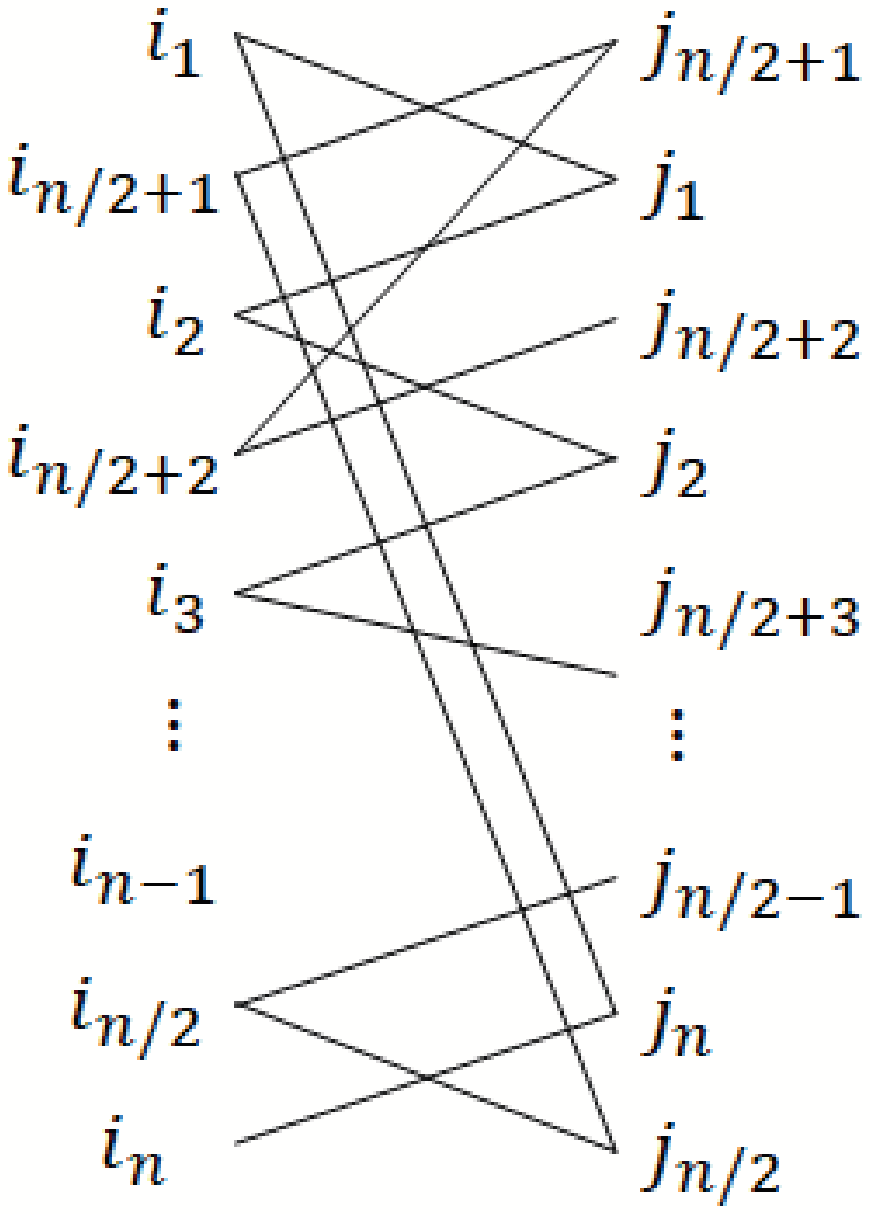}
		\end{subfigure}

	\end{figure}
	
	\begin{equation*}
		\pi^*=\begin{cases}
			\left(
			\begin{array}{ccccccccc}
				i_1,& i_{(n+3)/2},& i_2,& i_{(n+5)/2},&  \cdots,& i_{(n-1)/2},& i_n,& i_{(n+1)/2}\\
				j_{(n+1)/2},& j_1,& j_{(n+3)/2},& j_2,&  \cdots,& j_{n-1},& j_{(n-1)/2},& j_n
			\end{array}
			\right),& if\; n\; is\; odd
			\\
			\left(
			\begin{array}{ccccccccc}
				i_1,& i_{n/2+1},& i_2,& i_{n/2+2},& i_3,& \cdots,& i_{n-1},& i_{n/2},& i_n \\
				j_{n/2+1},& j_1,& j_{n/2+2},& j_2,& j_{n/2+3},& \cdots,& j_{n/2-1},& j_n,& j_{n/2}
			\end{array}
			\right),& if\; n\; is\; even
		\end{cases}
	\end{equation*}
	
	This permutation $\pi^*$ is exactly the one that makes the sum $\sum_{i=1}^n\sum_{j=1}^n T[i,j]C[i,\pi^*_j]$ equal to the weight of the optimal Hamilton cycle $HC^*$.
\end{proof}

Algorithm \ref{alg:drp-approx} gives an approximate algorithm for the DRP problem.

\begin{algorithm}
	\caption{Approximate algorithm for DRP}\label{alg:drp-approx}
	\KwIn{An instance of DRP, with matrix $T[i,j], C[i,j]$.}
	\KwOut{A permutation $\pi\in \Pi(n)$ which minimizes the cost.}
	Solve the linear assignment problem (LAP) with input $T[i,j]$\;
	return the permutation got by the LAP as the approximate solution to DRP\;
\end{algorithm}

The LAP in Algorithm \ref{alg:drp-approx} refers to the following problem: $\min_{\pi\in \Pi(n)} T[i,j]A[i,\pi_j]$ where $A[i,j]=1$ when $i\ne j$ and $A[i,i]=0$ $1\le i,j\le n$.

\begin{theorem}\label{thrm:drp-ratio}
	The approximation ratio of Algorithm \ref{alg:drp-approx} is bounded by  $c_{max}/c_{min}$ where $c_{max}=\max\limits_{1\le i,j\le n}C[i,j], c_{min}=\min\limits_{1\le i,j\le n}C[i,j]$.
\end{theorem}

\begin{proof}
	Denote $OPT$ as the cost of the optimal solution to DRP, which corresponds to a permutation $\pi^D$. Denote $APR$ as the cost of the solution given by Algorithm \ref{alg:drp-approx}, which corresponds to a permutation  $\pi^L$.

	Divide all elements in $C[i,j]$ by $c_{min}$, and the minimum non-zero element becomes $1$. Denote $A[i,j]$ be the matrix that $A[i,j]=1$ when $i\ne j$ and $A[i,i]=0$ $1\le i,j\le n$. Then we have
	
	$$OPT/c_{min}=\sum_{i=1}^n\sum_{j=1}^n T[i,j] C[i,\pi^D_j]/c_{min}\ge \sum_{i=1}^n\sum_{j=1}^n T[i,j] A[i,\pi^D_j] \ge \sum_{i=1}^n\sum_{j=1}^n T[i,j] A[i,\pi^L_j] $$
	
	The last $\ge$ is due to $\pi^D$ is a feasible solution to LAP and $\pi^L$ is the optimal solution to LAP, where LAP is the linear assignment problem with input $T[i,j]$.
	
	On the other hand, divide all elements in $C[i,j]$ by $c_{max}$, and the maximum non-zero element becomes $1$. Thus we have 
	
	$$APR/c_{max}=\sum_{i=1}^n\sum_{j=1}^n T[i,j] C[i,\pi^L_j]/c_{max}\le \sum_{i=1}^n\sum_{j=1}^n T[i,j] A[i,\pi^L_j]$$
	
	Hence, we proved that $OPT/c_{min}\ge ARP/c_{max}\Rightarrow APR/OPT\le c_{max}/c_{min}$.

\end{proof}

\begin{remark}
	The approximation ratio $c_{max}/c_{min}$ is not theoretically good enough. But racall that $c_{max}/c_{min}$ are the ratio between the maximum and minimum machine-to-machine communication cost inside a cluster. Thus this approximation ratio is actually small.
\end{remark}

\section{The general problem: minimizing both communication and IO costs} \label{sec:optimize-comm-io}
In this section we try to solve the general optimization problem on the EMPC model, which is to minimize the total of communication and IO cost.

Let $f_i(i), f_f(i)$ be the amount of data resides initially and finally in machine $i$ for $1\le i\le p$, respectively. Let $T(i,j)$ be the amount of data transferred from machine $i$ to $j$. Denote $C(i,j)$ to be the communication cost from machine $i$ to $j$, and $C_{IO}$ be the IO cost of the machines. Like mentioned in Section \ref{sec:model},  we assume $C_{IO}=1$  for simplicity. Let $F_{IO}(\cdot)$ be a problem-specific function to measure the amount of IO operations. We have the following objective function for the total cost of communication and IO.

$$\min_{f_f(i),\pi\in \Pi(p)}  \sum_{i=1}^{p}\sum_{j=1}^{p}\{T(i,j)C(i,\pi_j)\}+  \max_i F_{IO}(f_f(i))$$
subject to
$$ \sum_{j=1}^p{ T(j,i)}-\sum_{j=1}^p{T(i,j)}=f_f(i)-f_i(i)$$

The above description is not complete since there must be another constraint on $f_f(i)$ to ensure the correctness for a given problem. So it only provides an image for the complexity of optimizing the total cost. In this paper we consider a much simpler case on the Sorting problem. 

\subsection{The general optimization problem on Sorting}
For the case of sorting, the above general optimization problem can be simplified. 

\begin{definition}[General Optimization problem on Sorting, GOP-on-Sorting] \label{def:gop-on-sort}
	The input is a set $S$ of $n$ integers divided into $p$ subsets $S_1=\{s_{1,1},s_{1,2},\cdots, s_{1,n_1} \},\cdots, S_p=\{ s_{p,1},s_{p,2},\cdots, s_{p,n_p}\}$, where $\sum_{i=1}^{p}n_i=n$ and $p>1$. The output is to
	return $p-1$ integers $s^*_1,\cdots s^*_{p-1}\in S$ and a permutation $\pi\in \Pi_p$ such that the following objective function is minimized:
	$$\min_{\pi\in \Pi(p)}\sum_{i=1}^{p}\sum_{j=1}^{p} T[i,j]C[i,\pi_j]+\max_i{ L_i\log{L_i}}   $$
	
	where
	$T[i,j]=|S_i\cap (s^*_{j-1},s^*_j]|, L_i=|S\cap (s^*_{i-1},s^*_i]|$ and $s^*_0=-\infty, s^*_p=\infty$. See Figure \ref{fig:gop-on-sort} for a demonstration.
\end{definition}

\begin{figure}[H]
	\centering
	\includegraphics[width=0.3\textwidth]{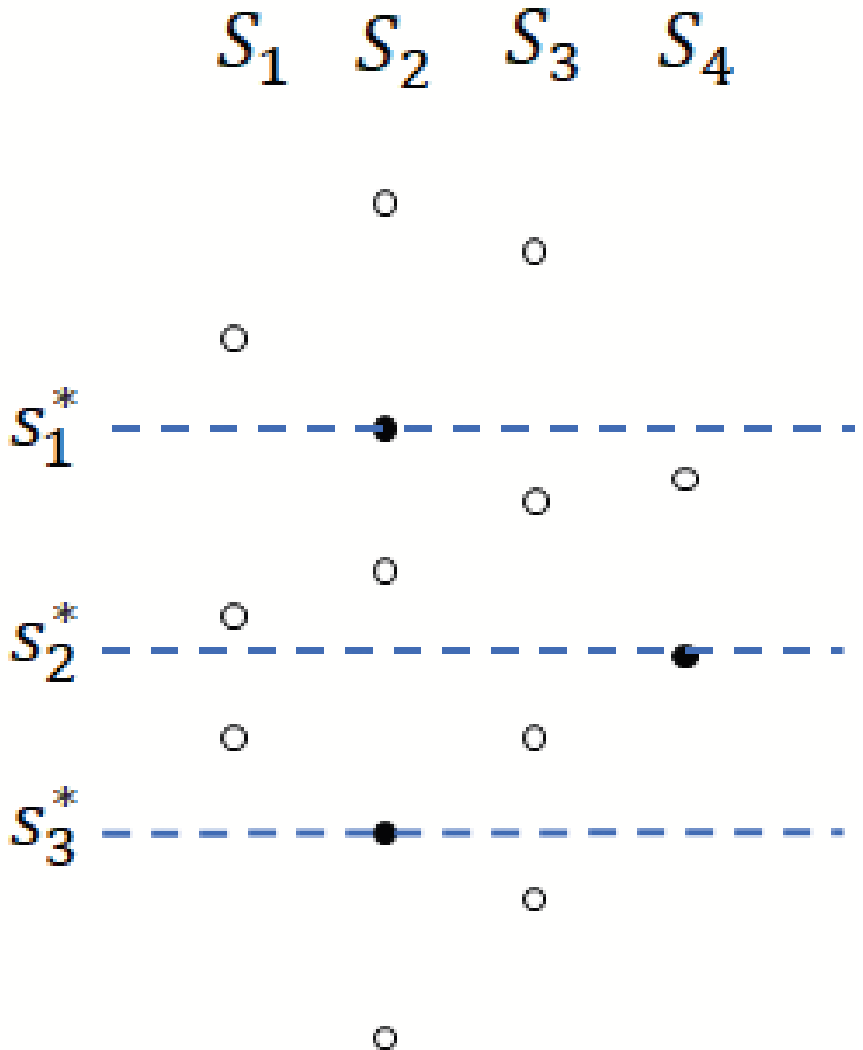}
	\caption{Demonstration of the general optimization problem on Sorting.}\label{fig:gop-on-sort}
\end{figure}

\begin{remark}
	In the last remark in Section \ref{subsec:model-power} we have mentioned that the communication and IO operations can be overlapped. However in Definition \ref{def:gop-on-sort} the total cost is modeled as the addition of communication and IO costs, somewhat ignoring the overlapping of them. We think this simplification is reasonable for a theoretical analysis for minimizing the total cost. The overlapping of the communication and IO operations can be taken into consideration while designing practical algorithms.
\end{remark}
\begin{theorem}
	The problem GOP-on-Sorting defined in Definition \ref{def:gop-on-sort} is in XP.	
\end{theorem}

\begin{proof}
	The problem can be solved as follows: first enumerate all possible $p-1$ splitters from the $n$ inputs, then scan the $n$ inputs to determine $T[i,j]$ and $L_i$, finally enumerate all possible permutations on $[1,p]$ to find the optimal solution. Henceforth, the algorithm takes $O(\tbinom{n}{p}\cdot n\cdot p!)=O(n^{p+1}p!)$ time. By the definition of XP class \cite{Downey1999parameterized}, this problem is in XP (Fixed-Parameter Tractable for Each Parameter).
\end{proof}

We give an simple but efficient approximate algorithm for the general optimization problem on Sorting.

\begin{algorithm}[H]
	\caption{Approximation algorithm}\label{alg:gop-on-sort-approx}
	\KwIn{An instance of the GOP-on-Sorting.}
	\KwOut{a set of $p-1$ splitters, and a permutation $\pi\in \Pi(p)$.}
	Sort the $n$ inputs into an ordered sequence $s_1,\cdots, s_n$\;
	Choose the splitters as $s_{\frac{n}{p}}, s_{2\frac{n}{p}}\cdots, s_{\frac{n}{p}(p-1)}$\;
	Determine $T[i,j]$ and $L_i$ according to the chosen splitters\;
	With input as $T[i,j]$, use Algorithm \ref{alg:drp-approx} to solve the DRP problem and return an permutation $\pi$\;
	Return the splitters and $\pi$ as the output\;
	
\end{algorithm}

\begin{theorem}\label{thrm:gop-ratio}
	If $p2^p\le n$ then the approximation ratio of Algorithm \ref{alg:gop-on-sort-approx} is bounded by $\max \{c_{max}/c_{min},2 \}$.
\end{theorem}

\begin{proof}
	
	Let $APR$ and $OPT$ denote cost of the approximate and optimal solution of GOP-on-Sorting, respectively.
	Notice that the objective function is the addition of two parts, one representing the communication cost and the other representing the IO cost. We then divide $APR$ and $OPT$ into two parts, i.e., $APR=APR^{(1)}+APR^{(2)}$ and $OPT=OPT^{(1)}+OPT^{(2)}$.
	Apparently, $APR^{(2)}=\frac{n}{p}\log{\frac{n}{p} }$ and $APR^{(2)}\le OPT^{(2)}$.
	
	Denote $AOPT$ as the approximate solution of the DRP with input obtained using the optimal splitters in the solution of OPT, and denote $r=c_{max}/c_{min}$. Then we have $AOPT\le r\cdot OPT^{(1)}$.
	
	Next we try to 	derive an addictive bound between $APR^{(1)}$ and $AOPT$. Since they are the cost of two instances of LAP (linear assignment problem) obtained using different splitters, we turn to bound the difference between the solution of two arbitrary instances of LAP.
	
	For given $n$ and $p$, for arbitrary two LAP problems obtained by different splitters on the same GOP-on-Sorting problem, the difference of the cost of them is at most $\frac{p-1}{p}n$ for the following reason. The LAP obtained by GOP-on-Sorting problem has the following form:
	\begin{align*}
		&\min_{\pi\in \Pi(p)} \quad T[i,j]C[i,\pi_j] \\
		& \begin{array}{r@{\quad}l@{}l@{\quad}l}
			s.t.& \sum_{i=1}^{p}\sum_{j=1}^{p}T[i,j]=n ,&\\
			&	C[i,j]=1, 1\le i\ne j\le p,&\\
			& C[i,j]=0, 1\le i=j\le p .&\\
		\end{array}
	\end{align*}

	It can be easily verified that: the minimum possible cost is 0, while the $T[i,j]$ matrix $T[i,j]$ is diagonal; and the maximum possible cost is $\frac{p-1}{p}n$, while the input is as $T[i,j]=\frac{1}{p^2}n$, $1\le i,j\le p$. 
	
	By the above analysis, we have the following inequalities.
	\begin{equation*}
		\begin{aligned}
			APR &= APR^{(1)}+APR^{(2)} \\
			&\le AOPT+DIFF+ APR^{(2)} \\
			&\le r\cdot OPT^{(1)}+ \frac{p-1}{p}n+ \frac{n}{p}\log{\frac{n}{p} }\\
			&\le r\cdot OPT^{(1)}+\frac{n}{p}\log{\frac{n}{p} } +\frac{n}{p}\log{\frac{n}{p} }\\
			&\le r\cdot OPT^{(1)}+2\cdot OPT^{(2)}\\
			&\le \max\{r,2\} \cdot (OPT^{(1)}+OPT^{(2)})\\
			&\le \max\{r,2\} \cdot OPT
		\end{aligned}
	\end{equation*}
	
	In the first step we use $DIFF$ to denote the difference between  $APR^{(1)}$ and $AOPT$, and we have proved that $DIFF\le\frac{p-1}{p}n$.
	In the third step we use the inequality $p-1\le \log{\frac{n}{p}}$, which is because of $p2^p\le n\Rightarrow p\le \log{\frac{n}{p}}$.
\end{proof}

\section{Conclusion}\label{sec:conc}
In this paper we proposed the EMPC model for parallel computation considering both IO cost and non-uniform communication cost. We declared and solved three new problems on EMPC model, which are the parallel IO-optimality, the data redistribution problem that minimizes the communication cost, and the general optimization problem that minimizes the total of communication and IO cost. We have shown the hardness of these problems and gave approximate algorithms for them. However, there are a lot of research problems that are not covered by this work, which are discussed in the next section as future works.

\section{Future works}\label{sec:fwork}
\subsection{Communication-optimality}
We have defined the optimality respect to the total IO cost in this paper. Note that parallel computation can be divided into three parts, i.e., computation, communication and IO operation. The optimality on parallel computation is considered by former researches, and the optimality on IO is considered in this work, leaving the optimality on communication unsolved. However, the situation on communication optimality is more complex, since there are a lot of parallel models and each of them leads to different communication complexity. However, we notice that many research works consider the communication \textit{round} complexity, but the communication \textit{amount} complexity is rarely studied. We believe that the communication amount complexity is problem-oriented and model-oblivious. Henceforth, it may be a right direction to use the communication amount complexity to define the communication optimality, and then fills the last puzzle of optimality in parallel computation.

\subsubsection{Data Duplication Problem and Data Pre-distribution Problem}
For the data redistribution problem, we only use sorting as an example. Actually the data redistribution problem represents the class of problems that needs only one round of communication without data duplication, which may be the simplest case. There are other kind of problems, e.g., that need one round of communication but with duplication, and that need multiple rounds of communication. Actually, the parallel join problem in database captures both of the two cases. The Hypercube algorithm  \cite{Afrati2010} is an one-round algorithm which needs to replicate the data tuples. The GYM  algorithm is  \cite{Afrati2017} a multi-round algorithm with only one round of data redistribution. We may define the following two problems: Data Duplication Problem and Data Pre-distribution Problem.

\begin{definition}[Data Duplication Problem]
	For a problem that needs one round of communication and duplication, find a communication plan that guarantees the correct output and  minimize the total communication cost.
\end{definition}

\begin{definition}[Data Pre-distribution Problem]
	For a problem that needs one round of data redistribution and multiple rounds of communication, find a way to pre-distribute the data so that the correct output is guaranteed and the total cost of the following rounds of communication is minimized.
\end{definition}

As we mentioned, both of the above problem can use the parallel join as the example case. Since the data redistribution problem is already  NP-complete, these two problems are even harder. We leave them as future works.

\subsubsection{Minimizing both communication and IO}
For the general optimization problem considering both communication and IO costs, we only considered the specialized  problem on sorting. There are two directions of future work. 
First, when applied to TeraSort and other realistic parallel sorting algorithms the given approximate Algorithm \ref{alg:gop-on-sort-approx} must run on the collected sample (Line 3 of Algorithm \ref{alg:terasort}). Since the sample is another approximation of the total data, the performance of the approximate algorithm respect to the total data must be analyzed too. Second, the general optimization problem can be also applied to other problems. As we have mentioned, the problem on sorting may be the simplest case since it only needs to choose $p-1$ splitters. Combined with the Data Duplicate Problem and Data Pre-Distribution Problem, it will be a great challenge to minimize the total communication and IO costs on other problems.



\bibliography{library}
\end{document}